\documentclass[a4paper,fleqn]{cas-dc}

\usepackage[authoryear,longnamesfirst]{natbib}
\usepackage{mathtools}
\usepackage{amsthm}
\usepackage{multirow}
\IfFileExists{stix.sty}{}{\usepackage{lmodern}}
\usepackage{microtype}
\usepackage{placeins}
\usepackage{xurl}

\newtheorem{theorem}{Theorem}
\newtheorem{proposition}{Proposition}
\newtheorem{lemma}{Lemma}
\theoremstyle{definition}

\newtheorem{assumption}{Assumption}
\newtheorem{remark}{Remark}

\newcommand{\R}{\mathbb{R}}
\newcommand{\indep}{\perp\!\!\!\perp}
\newcommand{\notindep}{\not\!\perp\!\!\!\perp}
\newcommand{\origcgc}{\textsf{c-GC}}
\newcommand{\origcgcstar}{\textsf{c-GC}$^\ast$}
\newcommand{\fastcgc}{\textsf{fastC-GC}}
\newcommand{\fastcgcstar}{\textsf{fastC-GC}$^\ast$}
\newcommand{\order}{\mathcal{O}}
\newcommand{\graphpanel}[3]{%
  \begin{minipage}[t]{0.31\textwidth}
    \vspace{0pt}%
    \centering
    \includegraphics[width=\linewidth]{#1}\par
    \smallskip\footnotesize\textbf{(#2)}~#3
  \end{minipage}%
}
\newcommand{\chamberkey}{%
  \begin{minipage}[c]{0.18\textwidth}
    \vspace{0pt}%
    \centering\footnotesize
    \textbf{Variable key}\par\smallskip
    \fontsize{7.5}{9}\selectfont
    \begin{tabular}{@{}rl@{}}
    1 & Hatch position\\
    2 & Input load\\
    3 & Output load\\
    4 & Potentiometer 1\\
    5 & Potentiometer 2\\
    6 & Input current\\
    7 & Output current\\
    8 & Downwind pressure\\
    9 & Upwind pressure\\
    10 & Input fan speed\\
    11 & Output fan speed\\
    12 & Microphone\\
    13 & Intake pressure\\
    14 & Ambient pressure\\
    15 & Signal 1\\
    16 & Signal 2
    \end{tabular}
  \end{minipage}%
}
\newcommand{\riverskey}{%
  \begin{minipage}[c]{0.18\textwidth}
    \centering\footnotesize
    \textbf{Station-ID key}\par\smallskip
    \begin{tabular}{@{}rl@{}}
    A & 0\\
    B & 94\\
    C & 503\\
    D & 504\\
    E & 505
    \end{tabular}
  \end{minipage}%
}
\makeatletter
\def\@thehead{}
\def\@thefoot{}
\makeatother

\begin{document}

\let\WriteBookmarks\relax
\renewcommand{\floatpagefraction}{0.85}
\renewcommand{\dblfloatpagefraction}{0.95}
\renewcommand{\dbltopfraction}{0.95}
\renewcommand{\textfraction}{0.05}
\setlength{\textfloatsep}{12pt plus 2pt minus 2pt}
\setlength{\dbltextfloatsep}{12pt plus 2pt minus 2pt}
\setlength{\dblfloatsep}{10pt plus 2pt minus 2pt}
\makeatletter
\setlength{\@fptop}{0pt}
\setlength{\@fpsep}{12pt}
\setlength{\@fpbot}{0pt plus 1fil}
\setlength{\@dblfptop}{0pt}
\setlength{\@dblfpsep}{12pt}
\setlength{\@dblfpbot}{0pt plus 1fil}
\makeatother
\raggedbottom

\shorttitle{Efficient RCCP-Guided Graph Estimation}
\shortauthors{S. A. Adedayo}

\title[mode=title]{\fastcgc{} and \fastcgcstar{}: Efficient
RCCP-Guided Graph Estimation for Time Series}

\author[1]{S. A. Adedayo}[orcid=0000-0003-2990-142X]
\ead{sadiq.adedayo@univie.ac.at}

\affiliation[1]{organization={Univie Doctoral School of Computer Science
(DoCS)},
  addressline={Wahringerstra{\ss}e 29},
  city={Vienna},
  postcode={1090},
  state={Vienna},
  country={Austria}}

\hypersetup{
  pdftitle={fastC-GC and fastC-GC*: Efficient RCCP-Guided Graph Estimation for Time Series},
  pdfauthor={S. A. Adedayo}}

\begin{abstract}
  We introduce \fastcgc{} and \fastcgcstar{}, computationally efficient methods
  for time-series graph estimation using the two-test rule motivated by
  Reichenbach's common cause principle (RCCP) and the conditioning sets of
  causalised Granger causality (\origcgc{} and \origcgcstar{}). The methods use
  standardized data, ridge-residual correlation, and complete circular shift
  enumeration evaluated by Fourier transforms. Shared precision calculations
  and Schur downdates avoid repeated regression fits. On a single-threaded
  benchmark, \fastcgc{} was 68--277 times faster than the compared \origcgc{}
  implementation, and \fastcgcstar{} was 192--886 times faster than
  \origcgcstar{}. Across four suites with reference graphs, aggregate mean
  area under the precision--recall curve (AUPRC) differed from the corresponding
  baseline by at most 0.019, but F1 was lower at the selected thresholds in
  every suite. The computational gains therefore accompany a change in
  estimator and selected graphs. Among eight methods, Bagged PCMCI+ had the
  highest mean AUPRC on sparse VAR, an autocorrelated structural process, and
  DREAM3; CUTS+ led on Lorenz--96. Wind-tunnel and river-discharge applications
  illustrate differences between inferred graphs without establishing causal
  accuracy. We prove the computational identities and a restricted population
  adjacency result for \fastcgcstar{}. Causal interpretation requires the
  stated graphical assumptions; validity of the implemented tests additionally
  requires a suitable shift-invariant null law. Empirical null rejection rates
  informed the conditional threshold \(\beta=0.005\), without establishing
  general test validity or graph-wide error control.
\end{abstract}

\begin{keywords}
causalised-GC \sep Reichenbach's Principle \sep circular shift calibration
\sep ridge residualization \sep time series \sep causal discovery
\end{keywords}

\maketitle

\hypersetup{
  pdftitle={fastC-GC and fastC-GC*: Efficient RCCP-Guided Graph Estimation for Time Series},
  pdfauthor={S. A. Adedayo}}

\section{Introduction}
\label{sec:introduction}

Granger causality (GC) asks whether the past of one process improves prediction of
another process beyond the information already available in the target's
history~\citep{granger1969}. This predictive definition is useful, but a
predictive edge is not automatically a causal edge. Common causes, omitted
variables, model misspecification, and conditioning on colliders can all change
the resulting graph.

Reichenbach's common cause principle (RCCP) motivates checking a proposed
relation both before and after adjustment for possible common causes.
\citet{adedayo2026} applied this reasoning to lagged source--target pairs: an
unconditional test asks whether the variables are dependent, and a conditional
test asks whether dependence remains after adjustment for a specified set of
observed states. A candidate edge is retained only when both tests support
dependence. The \origcgc{} and \origcgcstar{} procedures differ in the
conditioning set used for the second test.

We introduce \fastcgc{} and \fastcgcstar{} as standalone algorithms that use
the same RCCP decision rule and the \origcgc{} and \origcgcstar{} conditioning
sets, respectively. Their estimators are nevertheless different. The correlation-based
formulation standardizes the lag-expanded rows, applies a positive covariance
ridge, and calibrates dependence with complete circular shift orbits. Covariance
calculations, residual recovery, and shift evaluation are organized so that
many candidate edges can be processed together. These choices can change an
inferred graph, so agreement with \origcgc{} or \origcgcstar{} is an empirical
question rather than part of the algorithm definition.

The central empirical question is whether these computational changes retain
useful graph-recovery performance. Across the reference-graph suites, the fast
methods closely tracked the aggregate ranking performance of their c-GC
counterparts, but produced lower F1 at the selected thresholds. The runtime
results and this accuracy tradeoff are evaluated separately from the
population-level causal argument.

Other time-series methods reduce the cost and variance of large conditioning
sets in different ways. A broad comparison of GC measures found
that the dimension-reduced PMIME, RCGCI, and RGPDC procedures performed well,
particularly in high-dimensional settings~\citep{siggiridou2019}. CEDAR
recently combined lag screening with targeted conditional-independence tests
for sparse autoregressive processes~\citep{fesanghary2026}; it is included as
a recent lag-resolved constraint-based comparator. Variable-lag GC instead
addresses delays that change over time, whereas time-reversed GC is a
directional check designed to reject effects caused by
mixing and measurement noise~\citep{amornbunchornvej2021,winkler2016}. Those
two methods address different failure modes from the fixed-lag construction
studied here and are not included in the comparison.

The main contributions are:

\begin{itemize}
  \item the \fastcgc{} and \fastcgcstar{} algorithms, which combine the RCCP
  decision rule and established conditioning sets with standardized,
  regularized dependence estimation;
  \item vectorized algebraic reformulations for complete circular shift
  enumeration, ridge-residual recovery, and \fastcgc{} Schur downdates from
  one global precision matrix;
  \item proofs of the intersection--union decision rule and the algebraic
  identities used by the correlation-based formulation, a restricted population-level
  causal result for positive-lag \fastcgcstar{} edges, and an explanation of
  why that graph-theoretic result is not an
  estimator-consistency theorem;
  and
  \item paired score-level and graph-recovery comparisons with \origcgc{},
  \origcgcstar{}, CEDAR, and other external causal-structure methods. A
  separate null study informs threshold selection and documents the shift
  test's finite resolution. A single-thread timing grid compares
  end-to-end runtime with the two c-GC algorithms.
\end{itemize}

\section{Method}
\label{sec:method}

The two algorithms separate the RCCP decision rule and conditioning sets from
the dependence measure. Their marginal test uses Pearson correlation, and
their conditional test uses ridge-residual correlation. The precision,
residual-recovery, Schur, and complete-orbit derivations below define this
correlation-based formulation.

\subsection{Problem and lag-expanded representation}
\label{sec:notation}

Let \(X\in\R^{n\times T}\) contain \(n\) observed processes and \(T\)
consecutive samples, indexed from zero. In Equation~\eqref{eq:lag-expanded},
column slices include the lower endpoint and exclude the upper endpoint. Following the conditioning-depth notation of
\citet{adedayo2026conditioning}, \(p\) is both the number of represented past
states and the maximum candidate lag. Set \(L=T-p\), and form
the lag-expanded matrix
\begin{equation}
 Z =
 \begin{bmatrix}
 X_{\cdot,p:T}\\
 X_{\cdot,p-1:T-1}\\
 \vdots\\
 X_{\cdot,0:T-p}
 \end{bmatrix}
 \in\R^{m\times L},
 \qquad m=n(p+1).
\label{eq:lag-expanded}
\end{equation}
Row \(i=(a,\tau)\), where \(a=i\bmod n\) and
\(\tau=\lfloor i/n\rfloor\), represents process \(a\) at lag \(\tau\).
Targets are restricted to the first \(n\) rows, which represent the present
state.

Each row is centered and scaled once:
\begin{equation}
 \widetilde Z_{i,t}
 =\frac{Z_{i,t}-\overline Z_i}{s_i},
 \qquad
 s_i^2=\frac{1}{L}\sum_{t=1}^{L}
       (Z_{i,t}-\overline Z_i)^2.
\label{eq:standardize}
\end{equation}
A row with \(s_i\) below a floating-point tolerance is marked constant. Every
pair involving such a row receives score zero and \(p\)-value one. The use of
\(1/L\), rather than \(1/(L-1)\), only fixes the numerical scale; correlation
statistics are unchanged.

Table~\ref{tab:notation} summarizes the notation used below.

\begin{table}[t]
\centering
\caption{Core notation.}
\label{tab:notation}
\small
\begin{tabular}{@{}ll@{}}
\toprule
Symbol & Meaning\\
\midrule
\(n,T,p\) & variables, samples, and maximum lag\\
\(L=T-p\) & aligned sample length\\
\(m=n(p+1)\) & number of lag-expanded rows\\
\(i=(a,\tau)\) & process \(a\) at lag \(\tau\)\\
\(j=(b,0)\) & present-state target\\
\(\widetilde Z\) & centered, row-standardized expanded data\\
\(\lambda>0\) & dimensionless ridge on the correlation scale\\
\(\alpha,\beta\) & marginal and conditional decision thresholds\\
\(N_{\mathrm{eval}}\) & number of evaluated conditional pairs\\
\bottomrule
\end{tabular}
\end{table}

\subsection{Conditioning sets and the two-test rule}
\label{sec:two-test}

The definitions in this subsection state the RCCP decision rule and the two
conditioning sets in the notation used here. The rule and the distinction
between the \origcgc{} and \origcgcstar{} conditioning sets follow
\citet{adedayo2026}; the regularized estimators and computational path below
define \fastcgc{} and \fastcgcstar{}.

Let \(V=\{0,\ldots,m-1\}\) index the expanded rows. For a candidate
\(i=(a,\tau)\to j=(b,0)\), define the rows excluded before conditional
residualization as
\begin{align}
 F_{ij}^{\ast} &= \varnothing
 &&\text{for \fastcgcstar{},}
\label{eq:f-star}\\
 F_{ij} &=
 \{(a,\ell):0\leq \ell<\tau\}\setminus\{j\}
 &&\text{for \fastcgc{}.}
\label{eq:f-cgc}
\end{align}
The removal of \(j\) from \(F_{ij}\) is an endpoint ``add-back'' needed for
lagged self-edges. Put \(A_{ij}=V\setminus F_{ij}\) and
\(B_{ij}=A_{ij}\setminus\{i,j\}\). Thus \fastcgcstar{} adjusts for every
other expanded row. \fastcgc{} does not adjust for the later states of the
tested source.

The maximum lag \(p\) also determines the observed history available to the
conditional test. Before the pair-specific exclusions above, \(V\) contains
\(\{X_t,X_{t-1},\ldots,X_{t-p}\}\); when \(p>1\), the adjustment can therefore
include several past states. This distinction matters because omitted history
can leave residual dependence that appears as a direct edge. Conditioning-depth
diagnostics use graph changes under progressively deeper histories to assess
whether the observed state is adequate~\citep{adedayo2026conditioning}. In the
present algorithms, however, \(p\) controls both the oldest candidate lag and
the available conditioning history. This differs from the fixed-horizon
diagnostic in which the reported causal lag remains fixed while the
conditioning depth increases.

The marginal statistic is
\begin{equation}
 U_{ij}=
 \left|\operatorname{corr}(\widetilde Z_i,\widetilde Z_j)\right|.
\label{eq:unconditional-stat}
\end{equation}
For the conditional statistic, let \(e=(i,j)\), let
\(S=\widetilde Z\widetilde Z^\top/L\), and define the pair of ridge residuals
\begin{equation}
 R_e^{(ij)}
 =\widetilde Z_e
 -S_{eB}\bigl(S_{BB}+\lambda I\bigr)^{-1}\widetilde Z_B,
 \qquad B=B_{ij}.
\label{eq:ridge-residual}
\end{equation}
When \(B\) is empty, \(R_e^{(ij)}=\widetilde Z_e\). The final conditional
effect size for an evaluated candidate is
\begin{equation}
 C_{ij}=
 \left|\operatorname{corr}
 \left(R_i^{(ij)},R_j^{(ij)}\right)\right|.
\label{eq:conditional-stat}
\end{equation}
Both \(U_{ij}\) and \(C_{ij}\) are tested against circular shifts, as specified
in Section~\ref{sec:shift-test}.

Let \(p^U_{ij}\) and \(p^C_{ij}\) be the corresponding decision values. The
expanded edge mask is
\begin{equation}
 \widehat D_{ij}
 =
 \mathbf 1\{p^U_{ij}\leq\alpha\}
 \mathbf 1\{p^C_{ij}\leq\beta\}.
\label{eq:and-rule}
\end{equation}
Only \(p^U_{ij}\) and \(p^C_{ij}\) enter this decision. The precision quantity
\(q_{ij}^A\) and its analytic value \(p_{\mathrm{scr}}\), introduced below, are
recorded diagnostics and do not replace either shift-calibrated decision value.
Zero-lag self-pairs are excluded before testing; positive-lag autoregressive
edges remain eligible. For a requested set of lags \(\mathcal L\), the
aggregate adjacency and weight are
\begin{align}
 \widehat A_{ab}
 &=\max_{\tau\in\mathcal L}\widehat D_{(a,\tau),(b,0)},\\
 W_{ab}
 &=
 \begin{cases}
 \displaystyle
 \max_{\tau\in\mathcal L:\widehat D_{(a,\tau),(b,0)}=1}
 C_{(a,\tau),(b,0)}, & \widehat A_{ab}=1,\\
 0, & \widehat A_{ab}=0.
 \end{cases}
\label{eq:lag-aggregate}
\end{align}
The maximizing lag accompanies each aggregate edge.

\begin{proposition}[Intersection--union interpretation]
\label{prop:iut}
Suppose \(p^U\) and \(p^C\) are valid \(p\)-values for
\(H_0^U\) (marginal independence) and \(H_0^C\) (conditional independence).
For the composite null \(H_0=H_0^U\cup H_0^C\),
\begin{equation}
 p^{\mathrm{IUT}}=\max(p^U,p^C)
\label{eq:iut-p}
\end{equation}
is valid. More generally, the unequal-threshold rule
\(p^U\leq\alpha\) and \(p^C\leq\beta\) has size at most
\(\max(\alpha,\beta)\) over \(H_0\), not \(\alpha\beta\).
\end{proposition}

\begin{proof}
If \(H_0^U\) is true, then
\[
 \Pr\{\max(p^U,p^C)\leq u\}
 \leq \Pr\{p^U\leq u\}\leq u.
\]
The same argument applies when \(H_0^C\) is true. For unequal thresholds,
the rejection event is contained in \(\{p^U\leq\alpha\}\) under \(H_0^U\)
and in \(\{p^C\leq\beta\}\) under \(H_0^C\). Taking the worst case over the
union null yields the stated bound.
\end{proof}

This result establishes the logical form of the two-test rule. It does not
validate the component tests. Circular shift invariance and the linear
conditional-association model require separate assumptions below.

\subsection{Complete circular shift tests}
\label{sec:shift-test}

For centered \(x,y\in\R^L\), define the circular cross-correlation
\begin{equation}
 c_s=\sum_{t=0}^{L-1}x_t y_{(t-s)\bmod L},
 \qquad s=0,\ldots,L-1.
\label{eq:circular-correlation}
\end{equation}
The denominator \(\lVert x\rVert_2\lVert y\rVert_2\) does not depend on
\(s\). All shifted correlations can therefore be evaluated through real
Fourier transforms. The resulting decision value is
\begin{equation}
 p_{\mathrm{shift}}
 =\frac{
 1+\sum_{s=1}^{L-1}
 \mathbf 1\{|c_s|\geq |c_0|\}
 }{L}.
\label{eq:shift-p}
\end{equation}
The leading one in the numerator counts the identity shift. The reported value
therefore cannot be zero, consistent with the requirement that a randomization
\(p\)-value should not be reported as zero~\citep{phipson2010}.

\begin{theorem}[FFT evaluation of the circular orbit]
\label{thm:fft}
Let \(\mathcal F\) denote the length-\(L\) discrete Fourier transform. Up to a
possible reversal of the shift index,
\begin{equation}
 (c_0,\ldots,c_{L-1})
 =
 \mathcal F^{-1}\!
 \left(\mathcal F(x)\odot\overline{\mathcal F(y)}\right).
\label{eq:fft-orbit}
\end{equation}
Thus Equation~\eqref{eq:shift-p} equals brute-force enumeration of the complete
circular shift orbit up to floating-point error, and costs
\(\order(L\log L)\) per pair rather than \(\order(NL)\) for \(N\) sampled
shifts.
\end{theorem}

Proof: Appendix~\ref{app:algebra}.

\begin{remark}[Statistical validity is not an FFT identity]
\label{rem:shift-validity}
Theorem~\ref{thm:fft} is computational. Equation~\eqref{eq:shift-p} is a valid
randomization test only when the null law is invariant under the cyclic shifts
being enumerated. For the marginal test, that law concerns the standardized
row pair; for the conditional test, it concerns the two estimated ridge
residual series. Stationarity of the raw process alone does not establish either
invariance. Exact cyclic exchangeability is stronger than ordinary stationarity,
and wrap-around can be problematic for trends or boundary transients.
\citet{harris2020} gives conservative shift-test results under explicit
stationarity conditions for generic autocorrelated series. In this paper,
``complete'' or ``exact'' refers to enumeration of the chosen orbit, not to
automatic validity for every time series.
\end{remark}

The test has resolution \(1/L\). A conditional threshold \(\beta\) is
attainable only if \(L\geq\lceil1/\beta\rceil\); this condition must be checked
for every fitted dataset before thresholded graphs are interpreted.

\subsection{Ridge precision diagnostics and residual recovery}
\label{sec:ridge}

For any retained set \(A\), define
\begin{equation}
 M_A=S_{AA}+\lambda I,
 \qquad
 P^A=M_A^{-1},
 \qquad
 q_{ij}^{A}
 =\frac{|P^A_{ij}|}{\sqrt{P^A_{ii}P^A_{jj}}}.
\label{eq:precision-diagnostic}
\end{equation}
The sign convention for partial correlation is immaterial because the score is
absolute. The positive ridge makes \(M_A\) invertible even if \(m>L\). Related
shrinkage estimators have long been used to stabilize high-dimensional
covariance inversion~\citep{ledoit2004}; here \(\lambda\) is a fixed
dimensionless parameter rather than the Ledoit--Wolf estimate. We use
\(\lambda=10^{-6}\) after standardization. This value is neither data-selected
nor statistically optimal, and its use outside the reported setting requires
a calibration or sensitivity check.

\begin{lemma}[Existence and vanishing-ridge perturbation]
\label{lem:ridge}
For \(\lambda>0\), \(S+\lambda I\) is positive definite and
\(\|(S+\lambda I)^{-1}\|_2\leq1/\lambda\). If \(S\) is positive definite with
smallest eigenvalue \(\mu_{\min}>0\), then
\begin{equation}
 \left\|(S+\lambda I)^{-1}-S^{-1}\right\|_2
 \leq
 \frac{\lambda}
 {\mu_{\min}(\mu_{\min}+\lambda)}.
\label{eq:ridge-bound}
\end{equation}
\end{lemma}

Proof: Appendix~\ref{app:algebra}.

Fixed ridge regularization is not invariant to arbitrary measurement units if
it is applied before scaling. Standardization resolves that problem for the
absolute statistics used here.

\begin{proposition}[Affine scale invariance]
\label{prop:scale}
In exact arithmetic, replace every nonconstant input process by
\(X'_a=c_aX_a+d_a\), where \(c_a\neq0\). Provided this transformation does not
change which rows are classified as constant, the marginal scores, regularized
precision diagnostics, ridge-residual scores, and complete-orbit \(p\)-values
are unchanged after Equation~\eqref{eq:standardize}.
\end{proposition}

Proof: Appendix~\ref{app:algebra}.

The qualification about constant-row classification matters in finite
precision: extreme offsets or scales can cross a magnitude-dependent tolerance
even when the exact-arithmetic variance is nonzero.

The optional Student-\(t\) transform of \(q_{ij}^A\) is described in
Appendix~\ref{app:diagnostics}. It does not select candidates or replace the
circular shift decision values.

The next identity recovers the ridge residual time series without fitting each
regression separately.

\begin{theorem}[Ridge-residual block identity]
\label{thm:residual}
Partition a retained set \(A\) into tested endpoints \(e=(i,j)\) and controls
\(B=A\setminus e\). Let
\(P^A=(S_{AA}+\lambda I)^{-1}\). The residuals in
Equation~\eqref{eq:ridge-residual} satisfy
\begin{equation}
 R_e
 =
 \bigl(P^A_{ee}\bigr)^{-1}
 \bigl(P^A\widetilde Z_A\bigr)_e.
\label{eq:residual-identity}
\end{equation}
\end{theorem}

\begin{proof}
Block inversion gives
\[
 P^A_{eB}
 =-P^A_{ee}S_{eB}(S_{BB}+\lambda I)^{-1}.
\]
It follows that
\[
\begin{aligned}
 (P^A\widetilde Z_A)_e
 &=P^A_{ee}\widetilde Z_e+P^A_{eB}\widetilde Z_B\\
 &=P^A_{ee}\left[
 \widetilde Z_e-S_{eB}(S_{BB}+\lambda I)^{-1}\widetilde Z_B
 \right]\\
 &=P^A_{ee}R_e.
\end{aligned}
\]
The principal block \(P^A_{ee}\) is positive definite and therefore
invertible.
\end{proof}

For \(\lambda=0\), the correlation of ordinary least-squares residuals equals
the familiar precision partial correlation. For \(\lambda>0\), the empirical
covariance of \(R_e\) is generally not \((P^A_{ee})^{-1}\). Consequently
\(q_{ij}^{A}\) and \(C_{ij}\) need not coincide. Both fast algorithms retain
the former as a diagnostic and use the latter as the conditional statistic.
For a multivariate Gaussian law, zero ordinary partial correlation is
equivalent to conditional independence, but this equivalence does not hold for
general distributions~\citep{baba2004}. A positive ridge also breaks the exact
Gaussian bridge, as the following example shows.

\begin{remark}[Fixed ridge is not an exact conditional-independence criterion]
\label{rem:ridge-ci}
Let \(Z,\epsilon_X,\epsilon_Y\) be independent standard Gaussian variables and
set \(X=Z+\epsilon_X\) and \(Y=Z+\epsilon_Y\). Then
\(X\indep Y\mid Z\). After standardizing \(X\) and \(Y\), ridge residualization
on \(Z\) gives
\[
 R_X=\frac{\epsilon_X}{\sqrt{2}}
     +\frac{\lambda Z}{\sqrt{2}(1+\lambda)},
 \qquad
 R_Y=\frac{\epsilon_Y}{\sqrt{2}}
     +\frac{\lambda Z}{\sqrt{2}(1+\lambda)}.
\]
Hence
\[
 \operatorname{corr}(R_X,R_Y)
 =\frac{\lambda^2}{(1+\lambda)^2+\lambda^2}>0
 \qquad(\lambda>0).
\]
Thus linear Gaussianity alone does not make the implemented fixed-ridge
statistic an exact criterion for conditional independence. A small ridge can be
a useful numerical approximation when the unregularized covariance is
well-conditioned, as bounded by Lemma~\ref{lem:ridge}; that lemma does not cover
a singular covariance. An estimator-consistency result would require further
conditions, such as a vanishing ridge sequence and control of estimation and
randomization error. No such result is claimed here.
\end{remark}

\subsection{One global precision and \fastcgc{} Schur downdates}
\label{sec:schur}

\fastcgcstar{} has \(F_{ij}^{\ast}=\varnothing\), so it uses one global
precision matrix \(Q=(S+\lambda I)^{-1}\) for every pair. The \fastcgc{}
exclusions vary with source and lag, but their precision matrices can still be
computed from the same \(Q\).

\begin{theorem}[Precision after deleting a principal set]
\label{thm:schur}
Let \(M=S+\lambda I\), \(Q=M^{-1}\), and partition the indices into retained
rows \(A\) and excluded rows \(F\). Then
\begin{equation}
 \boxed{
 M_{AA}^{-1}
 =Q_{AA}-Q_{AF}Q_{FF}^{-1}Q_{FA}.
 }
\label{eq:schur-downdate}
\end{equation}
\end{theorem}

\begin{proof}
Write \(Q\) in \(A,F\) block form. The block-inverse formula says that the
\(AA\) block of \(Q^{-1}=M\) is
\[
 M_{AA}
 =
 \left(Q_{AA}-Q_{AF}Q_{FF}^{-1}Q_{FA}\right)^{-1}.
\]
All principal blocks are positive definite because \(M\) and \(Q\) are
positive definite. Inverting both sides yields
Equation~\eqref{eq:schur-downdate}.
\end{proof}

For a source \(a\) at lag \(\tau\), ordinary targets share the exclusion
signature \(\{(a,0),\ldots,(a,\tau-1)\}\); the lagged self-target uses the
same signature with \((a,0)\) added back. The algorithm caches each
dense downdated precision by signature. There are \(\order(np)\) signatures
(at most \(1+n(2p-1)\) for \(p\geq1\) with the stated self-edge policy), and every
downdate solves a system of size at most \(p\). Residuals for the evaluated
pairs are then recovered in batches with Theorem~\ref{thm:residual}.

\subsection{Algorithm}
\label{sec:algorithm}

Both fast algorithms perform the following steps.

\begin{enumerate}
  \item Construct \(Z\), standardize valid rows, and build the allowed-pair
  mask once.
  \item Compute all marginal scores and complete circular shift values by one
  row-wise real FFT followed by one batched inverse transform per target.
  \item Form \(S+\lambda I\) and obtain its symmetric precision through
  Cholesky solves.
  \item For \fastcgcstar{}, use the global precision. For \fastcgc{},
  compute each required precision from Equation~\eqref{eq:schur-downdate}.
  \item Recover ridge residuals for every evaluated pair in
  bounded batches, then record \(C_{ij}\) and \(p_{\mathrm{shift}}\) as its
  conditional score and decision value.
  \item Apply Equation~\eqref{eq:and-rule}, aggregate requested lags, and
  weight each retained edge by the strongest retained conditional score.
\end{enumerate}

Every allowed pair receives both dependence evaluations. Optional diagnostics
are described in Appendix~\ref{app:diagnostics}. The graph-recovery
comparison uses this complete evaluation throughout.

\subsection{Computational complexity}
\label{sec:complexity}

Let \(N\) be the baseline's sampled-shift count, let
\(N_{\mathrm{eval}}\leq mn\) be the number of evaluated conditional pairs, and
let \(\mathcal F\) contain the distinct \fastcgc{} exclusion signatures. For
\(F\in\mathcal F\),
\(|F|\leq p\).
Write
\[
 H_{\mathrm{Schur}}
 =\sum_{F\in\mathcal F}\left(m^2|F|+|F|^3\right),
 \qquad K=m+\log L.
\]
Table~\ref{tab:complexity} lists leading dense costs.

\begin{table*}[t]
\centering
\caption{Leading time complexities. The fast algorithms evaluate a complete
shift orbit, whereas the candidate-wise c-GC calculation samples \(N\) shifts.}
\label{tab:complexity}
\small
\begin{tabular}{@{}p{0.20\textwidth}p{0.32\textwidth}p{0.39\textwidth}@{}}
\toprule
Pass & Candidate-wise c-GC & Fast algorithms\\
\midrule
Marginal
&
\(\order(m^2L+mnNL)\)
&
\(\order(nmL\log L)\)
\\[3pt]
Conditional \fastcgcstar{}
&
\(\order(mn(Lm^2+m^3+NL))\)
&
\(\order(m^2L+m^3+N_{\mathrm{eval}}L\log L)\)
\\[3pt]
Conditional \fastcgc{}
&
\(\order(mn(Lm^2+m^3+NL))\)
&
\(\order\!\left(
m^2L+m^3+H_{\mathrm{Schur}}
+N_{\mathrm{eval}}LK\right)\)
\\
\bottomrule
\end{tabular}
\end{table*}

The marginal operation count decreases by approximately \(N/\log L\) when the
shift loop dominates. The \fastcgcstar{} path replaces \(mn\) least-squares problems by
one covariance factorization and batched \(2\times2\) solves. The
\fastcgc{} Schur formula removes repeated \(m\times m\) inversions, although
forming residual rows still costs \(\order(N_{\mathrm{eval}}mL)\), where
\(N_{\mathrm{eval}}\) is the number of allowed pairs.

Dense covariance storage remains \(\order(m^2)\). The \fastcgc{}
cache can store \(\order(|\mathcal F|m^2)\) numbers because it materializes a
dense downdated precision for every signature. Selected-entry Schur updates
would reduce this memory cost. When \(m\gg L\), a Woodbury formulation could
replace the \(m\times m\) factorization by an \(L\times L\) solve.

\subsection{Limits of the causal interpretation}
\label{sec:causal}

The algebra above does not turn association into causation. The next result is
stated for population conditional independence, not for the finite-sample
ridge statistic.

\begin{assumption}[Causal Markov condition]
\label{ass:causal-markov}
The joint law of the represented variables is Markov to a time-unrolled DAG
whose vertex set is the expanded window.
\end{assumption}

\begin{assumption}[Causal sufficiency]
\label{ass:causal-sufficiency}
The represented variables are causally sufficient within the expanded window.
In particular, an unrepresented ancestor outside the window does not act as a
common cause of represented nodes.
\end{assumption}

\begin{assumption}[Faithfulness]
\label{ass:faithfulness}
The joint law of the represented variables is faithful to the time-unrolled
DAG.
\end{assumption}

\begin{assumption}[Temporal scope]
\label{ass:temporal-scope}
The maximum causal lag is no larger than \(p\), and every edge has strictly
positive lag.
\end{assumption}

\begin{assumption}[Present-state sink]
\label{ass:present-sink}
The present target has no descendants in the expanded window.
\end{assumption}

These assumptions refine those stated for causalised-GC in the source
paper~\citep{adedayo2026}. \citet{pearl2009} provides standard definitions
of the causal Markov condition, faithfulness, and graphical separation. The
shift-invariance requirement in Remark~\ref{rem:shift-validity} is a separate
condition on the component tests.

\begin{theorem}[\fastcgcstar{} adjacency for a present sink at the population level]
\label{thm:adjacency}
Under Assumptions~\ref{ass:causal-markov}--\ref{ass:present-sink}, let
\(U=X^a_{t-\tau}\) with \(\tau\geq1\), and let \(Y=X^b_t\) be a
present-state target distinct from \(U\). Then \(U\) and \(Y\) are adjacent in
the observed time-unrolled DAG if and only if
\begin{equation}
 U\notindep Y\mid V\setminus\{U,Y\}.
\label{eq:all-other-adjacency}
\end{equation}
If they are adjacent, faithfulness also implies \(U\notindep Y\) marginally.
Thus, when the marginal and conditional independence relations are known
exactly, both \fastcgcstar{} tests retain these positive-lag adjacencies.
\end{theorem}

\begin{proof}
A direct edge \(U\to Y\) is an active path with no internal nodes, so
conditioning on nodes other than its endpoints cannot block it. Faithfulness
then gives conditional and marginal dependence.

Now suppose \(U\) and \(Y\) are not adjacent and condition on every other
observed node. Any path with a noncollider among its internal nodes is blocked,
because that noncollider is conditioned on. An active non-direct path would
therefore need every internal node to be a collider. In a DAG, two consecutive
internal nodes cannot both be colliders on the edge joining them. Hence the
only possible such path has length two and the form
\(U\to K\leftarrow Y\). It would require \(K\) to be a child of \(Y\).
The present target is a sink in the observed window, so no such \(K\) exists.
All non-direct paths are blocked, which proves conditional independence for
nonadjacent pairs by the causal Markov condition.
\end{proof}

This result explains why conditioning on later source states can preserve
adjacency for a present-time sink under Assumptions~\ref{ass:causal-markov}--\ref{ass:present-sink}.
It assumes an oracle for population independence. The implemented estimator
instead uses Pearson and fixed-ridge residual correlations; its consistency
does not follow from the theorem. Remarks~\ref{rem:ridge-ci} and
\ref{rem:shift-validity} give the separate limitations of the residual
criterion and shift reference distribution.

The \fastcgc{} conditioning set does not yield a similar
direct-adjacency theorem. By excluding later states of the source,
\fastcgc{} can leave a directed mediated path open. Its output should
therefore be described as a lagged predictive or ancestral graph unless a
separate argument establishes separation for the data-generating class.
Finally, graph recovery tests roughly \(n^2p\) hypotheses. Fixed edgewise
thresholds do not control a graph-wide error rate; dependency-robust
false-discovery procedures are available~\citep{benjamini2001}, but neither
linear variant applies them.

\section{Experiments}
\label{sec:experiments}

\subsection{Study design and validation}
\label{sec:study-design}

\subsubsection{Evaluation design}
\label{sec:evaluation-design}

The comparison includes \origcgc{} and \origcgcstar{}~\citep{adedayo2026}, PCMCI+
\citep{runge2020}, Bagged PCMCI+~\citep{debeire2024}, CEDAR
\citep{fesanghary2026}, and CUTS+~\citep{cheng2024}. The two c-GC algorithms
are direct baselines for their fast counterparts; identical results are not
assumed.

PCMCI+ supplies a lag-resolved constraint-based comparison designed for
autocorrelated time series. Its reported configuration uses partial correlation and positive
lags. Bagged PCMCI+ uses 200 bootstrap resamples to assess whether aggregation
improves that base procedure in short, high-dimensional, or strongly
autocorrelated samples. CEDAR retains at most one selected cross-lag for an
ordered variable pair, which limits the
exact-lag comparison when a pair is linked at several lags. CUTS+ adds a neural
predictive method designed for large and irregularly sampled systems. It
returns an aggregate variable graph and is therefore evaluated only on the
any-lag estimand.

Every method receives the same observation matrix for a given dataset. The
primary graph contains directed positive-lag cross-variable
edges. Lag-resolved results use a common lag--source--target convention, and
aggregate any-lag recovery is reported separately. For \origcgc{} and
\origcgcstar{}, applying the recorded thresholds to the lag-specific marginal
and conditional \(p\)-values and taking the logical OR across lags reproduced
the aggregate graphs. Both algorithms can therefore enter the exact-lag
comparison. CEDAR's final selection is binary, whereas its available
\(p\)-values describe preliminary stages before MCI pruning. Those values do
not rank the final candidate edges. CEDAR consequently contributes binary
exact-lag and any-lag metrics, but not AUROC or AUPRC.

The evaluation excludes self-edges and contemporaneous edges. With \(n\)
variables and \(p\) candidate positive lags, the exact-lag candidate set has
\(n(n-1)p\) entries and the aggregate set has \(n(n-1)\) ordered pairs.
The aggregate binary graph contains an edge if any eligible lag is selected.
For the four RCCP algorithms, the lag-specific ranking score is
\(1-\max(p^U,p^C)\); this score is distinct from the retained-edge effect
size in Equation~\eqref{eq:lag-aggregate}.

Area under the precision--recall curve (AUPRC) is the primary ranking metric;
edge prevalence provides a reference baseline. Area under the receiver
operating characteristic curve (AUROC) is secondary. The main recovery tables
report AUPRC and F1; the null analysis reports cross-edge false-positive rates
and the probability of at least one false edge. Ranking summaries contain the
seven methods with continuous final-edge scores, and binary summaries contain
all eight methods. The cross-link-free SCM cell is excluded from AUPRC
summaries and included in F1 summaries.

The seed is the unit of replication in the simulated suites; the five DREAM3
networks are distinct systems. Paired aggregate-AUPRC contrasts use two-sided
Wilcoxon signed-rank tests and paired 10,000-draw bootstrap intervals, with
Holm correction over eight contrasts within each suite. The eight comparisons
pair each fast method with its corresponding c-GC baseline, PCMCI+, Bagged
PCMCI+, and CUTS+. Individual edges are not treated as independent replicates.
Detailed numerical comparisons are retained in Appendix~\ref{app:comparisons}.

\begin{table*}[t]
\centering
\caption{Method characteristics and settings specified in this study.
The common scoring mask excludes self-edges and contemporaneous edges.
Dataset-specific maximum lags are described in the corresponding subsections.}
\label{tab:method-settings}
\footnotesize
\begin{tabular}{@{}p{0.24\textwidth}p{0.33\textwidth}p{0.35\textwidth}@{}}
\toprule
Methods & Dependence calculation or setting & Evaluated output\\
\midrule
\origcgc{}, \origcgcstar{} & Two-test rule; \(\alpha=0.01\), \(\beta=0.005\)
& Exact-lag and aggregate scores and binary graphs\\
\fastcgc{}, \fastcgcstar{} & Pearson/ridge-residual correlation; \(\lambda=10^{-6}\);
complete shifts; same thresholds & Exact-lag and aggregate scores and binary graphs\\
PCMCI+ & Partial correlation; positive candidate lags & Exact-lag and aggregate scores and binary graphs\\
Bagged PCMCI+ & 200 bootstrap resamples; directed-edge frequency & Exact-lag and aggregate scores and binary graphs\\
CEDAR & At most one selected cross-lag per ordered pair & Binary exact-lag and aggregate graphs\\
CUTS+ & Neural predictive method & Aggregate scores and binary graphs\\
\bottomrule
\end{tabular}
\end{table*}

The validation study keeps threshold calibration, ridge sensitivity, and
maximum-lag sensitivity separate. Null calibration uses independent i.i.d.,
stationary autoregressive, and deliberately drifting series; the drift arm is
an assumption-violation stress test. Ridge sensitivity crosses the stationary
null with a sparse-VAR alternative, and lag sensitivity fits orders below, at,
and above the generating order. Each cell contains 50 independent seeds. The
marginal and conditional \(p\)-values are retained separately so that the
declared \(\beta\) grid can be scored without refitting. The grid contains
2,000 fits. Validation intervals use 2,000 percentile-bootstrap resamples of
the independent seeds; edges are not treated as replicates.

Two properties of the recorded scores require explicit handling. First,
Equation~\eqref{eq:shift-p} reports \(\widehat p=(b+1)/L\), where \(b\)
counts nonidentity shifts with an absolute statistic at least as large as the
observed value and \(L=T-p\) is the aligned sample length. Here \(L\) is
not an autocorrelation-adjusted effective sample size. Since \(L<1000\)
for every controlled or in-silico series, the original conditional threshold
\(\beta=0.001\) cannot retain a fast-method edge. Based on the empirical
null study, we selected \((\alpha,\beta)=(0.01,0.005)\) before rescoring the
benchmark suites. The tables use this rule for all four RCCP algorithms;
changing these thresholds does not change their stored ranking scores.

Second, Bagged PCMCI+ reports the frequency of the modal graph mark, which need
not be an edge. The observed positive-lag marks were limited to a directed edge
and no edge. Directed-edge frequency is therefore the modal frequency when the
edge mark is modal and its complement when no edge is modal.

\subsubsection{Threshold calibration and sensitivity}
\label{sec:calibration-sensitivity}

Table~\ref{tab:null-calibration} reports empirical null rejection rates at
\(\alpha=0.01\) and the selected \(\beta=0.005\). Under the i.i.d. null, the mean cross-edge
false-positive rate was 0.0035--0.0045 for both algorithms across the three
series lengths; the upper endpoints of the seed-bootstrap 95\% intervals were
at most 0.0075. The stationary autoregressive null was more conservative:
0.0010--0.0015 for \fastcgc{} and 0--0.0005 for \fastcgcstar{}. The drifting
series do not satisfy the shift-invariance condition. In that stress test,
\fastcgc{} reached 0.0095 at 1000 observations (95\% interval
\([0.0055,0.0140]\)), whereas \fastcgcstar{} reached 0.0020. These are empirical rates for the specified settings. In particular, an
upper interval endpoint of 0.0075 does not establish a 0.005 rejection bound.
The study does not establish validity of the conditional test under all
configurations in the composite null, including marginally dependent pairs
that become independent after conditioning. Proposition~\ref{prop:iut}
would give a bound of \(\max(0.01,0.005)=0.01\) only if both component
values were valid \(p\)-values.

\begin{table*}[t]
\centering
\caption{Empirical null rejection at \(\alpha=0.01\) and \(\beta=0.005\). Entries are
mean cross-edge false-positive rates over 50 independent seeds at each series
length. The last column gives the range of graph-level probabilities of at
least one false edge over the three lengths; no graph-wide correction was
applied.}
\label{tab:null-calibration}
\footnotesize
\begin{tabular}{@{}llcccc@{}}
\toprule
Method & Null regime & \(T=250\) & \(T=500\) & \(T=1000\) & Graph error range\\
\midrule
\multirow{3}{*}{\fastcgc{}}
 & i.i.d. & .0045 & .0035 & .0040 & .14--.18\\
 & stationary AR & .0015 & .0010 & .0015 & .04--.06\\
 & drifting stress test & .0030 & .0040 & .0095 & .12--.30\\
\midrule
\multirow{3}{*}{\fastcgcstar{}}
 & i.i.d. & .0045 & .0035 & .0040 & .14--.18\\
 & stationary AR & .0005 & .0000 & .0000 & .00--.02\\
 & drifting stress test & .0000 & .0005 & .0020 & .00--.08\\
\bottomrule
\end{tabular}
\end{table*}

The nonzero graph-level error in Table~\ref{tab:null-calibration} is expected
when many edges are tested without multiplicity correction. Thus the selected
threshold controls neither familywise error nor false discovery rate. Among the three examined thresholds, \(\beta=0.001\) was unattainable,
\(\beta=0.005\) produced the rates in Table~\ref{tab:null-calibration}, and
\(\beta=0.01\) could only increase or leave unchanged the rejection set for
the same stored scores. We selected the middle value on this empirical basis.
The reported intervals describe the selection study; they do not constitute
an independent, post-selection validation of the threshold.

The ridge sensitivity analysis used
\[
 \lambda\in\{10^{-8},10^{-6},10^{-4},10^{-2}\}.
\]
Across this range, sparse-VAR AUPRC was 0.9583--0.9587 for \fastcgc{} and
0.9787 at the reported precision for \fastcgcstar{}. The corresponding
stationary-null false-positive rates remained 0.0010 and 0.0000. At the true
maximum lag of two, \fastcgc{} obtained AUPRC 0.9583 (95\% interval
\([0.9428,0.9719]\)) and F1 0.8858 (\([0.8546,0.9134]\)). \fastcgcstar{}
obtained AUPRC 0.9787 (\([0.9665,0.9886]\)) and F1 0.8974
(\([0.8732,0.9215]\)). Fitting four lags reduced these values to 0.9315 and
0.8496 for \fastcgc{}, and to 0.9537 and 0.8714 for \fastcgcstar{}. The
one-lag cells score only represented lags and omit a mean 2.84 true edges, so
they diagnose underfitting rather than complete-graph recovery.

\subsubsection{Wall-clock comparison}
\label{sec:runtime-comparison}

The direct runtime comparison used one Apple M5 Max host with thread pools
limited to one. Each method ran in a fresh process on the same
simulated series within a design cell. After one untimed warm-up, five repeats
timed fitting and graph construction together. The grid crossed 5 or 20
variables, 250 or 1000 observations, and fitted order 1 or 5. \origcgc{} variants
used 1000 permutations, whereas \fastcgc{} variants used their complete
shift computation. The measurements therefore compare the end-to-end
algorithms, not an isolated linear-algebra kernel.

\begin{table*}[t]
\centering
\caption{Single-thread wall-clock time in seconds. Each method entry is mean
(sample SD) across five timed repeats; the speedup is the conventional mean
divided by the corresponding fast mean. Repeats quantify timing noise on one
generated series per design cell.}
\label{tab:runtime}
\scriptsize
\begin{tabular}{@{}rrrccrccr@{}}
\toprule
\(n\) & \(T\) & \(p\) & \origcgc{} & \fastcgc{} & Speedup & \origcgcstar{} & \fastcgcstar{} & Speedup\\
\midrule
\multirow{4}{*}{5}
 & 250  & 1 & .1512 (.0008) & .00221 (.00012) & 68.4 & .1508 (.0005) & .00063 (.00004) & 238.6\\
 & 250  & 5 & .4787 (.0078) & .00700 (.00012) & 68.4 & .4718 (.0017) & .00116 (.00004) & 407.4\\
 & 1000 & 1 & .4533 (.0010) & .00289 (.00009) & 157.0 & .4547 (.0012) & .00137 (.00005) & 331.6\\
 & 1000 & 5 & 1.4132 (.0016) & .01397 (.00020) & 101.2 & 1.4241 (.0044) & .00742 (.00014) & 191.9\\
\midrule
\multirow{4}{*}{20}
 & 250  & 1 & 2.6851 (.0155) & .01343 (.00022) & 200.0 & 2.6845 (.0112) & .00729 (.00005) & 368.4\\
 & 250  & 5 & 13.7736 (.0266) & .05244 (.00053) & 262.6 & 13.9775 (.0092) & .01577 (.00030) & 886.3\\
 & 1000 & 1 & 7.8967 (.0291) & .02849 (.00056) & 277.2 & 7.9284 (.0123) & .02019 (.00027) & 392.7\\
 & 1000 & 5 & 35.5170 (.0610) & .15745 (.00099) & 225.6 & 35.9451 (.0220) & .11434 (.00128) & 314.4\\
\bottomrule
\end{tabular}
\end{table*}

Table~\ref{tab:runtime} shows that, across the eight cells, the speedup range
was 68.4--277.2 for \fastcgc{} and
191.9--886.3 for \fastcgcstar{}; the geometric means were 149.1 and 354.5.
These values are specific to the compared algorithms, grid, and host. The five
repeats within a cell are not independent
datasets, so no significance test is attached to the timing ratios.

\subsection{Results}
\label{sec:results}

\subsubsection{Stable sparse VAR}
\label{sec:scenario-sparse-var}

The first controlled process is a stable sparse vector autoregression,
\begin{equation}
 X_t=\sum_{\ell=1}^{p}A_{\ell}X_{t-\ell}+\varepsilon_t.
\label{eq:comparison-var}
\end{equation}
At lag one, each variable receives an autoregressive coefficient drawn
uniformly from \([0.20,0.40]\). Cross-variable support is sampled independently
at probability 0.12 over all lags, subject to at least one cross-variable
parent per target. Nonzero cross coefficients have magnitude drawn uniformly
from \([0.12,0.35]\) and an independently sampled sign. The coefficient
matrices are scaled together when necessary to keep the companion spectral
radius at or below 0.8.

After 500 burn-in observations, the experiment retains 250, 500, or 1000
samples from systems with 5, 10, or 20 variables and order 1, 2, or 5.
Innovations are standard Gaussian or zero-mean unit-variance Laplace. The full
factorial grid contains 50 seeds per cell, beginning at seed 10,000, for 2700
systems with exact lag-resolved truth.

\begin{table*}[t]
\centering
\caption{Sparse-VAR recovery, reported as mean (sample SD) across 2700
graphs. Binary metrics for the four RCCP algorithms use the selected
\((\alpha,\beta)=(0.01,0.005)\) rule. Dashes mark an unavailable metric or
estimand.}
\label{tab:sparse-var-recovery}
\footnotesize
\begin{tabular}{@{}lcccc@{}}
\toprule
Method & Aggregate AUPRC & Aggregate F1 & Exact-lag AUPRC & Exact-lag F1\\
\midrule
\origcgc{} & .952 (.058) & .873 (.124) & .886 (.104) & .832 (.135)\\
\fastcgc{} & .953 (.058) & .842 (.154) & .902 (.100) & .806 (.163)\\
\origcgcstar{} & .952 (.058) & .863 (.130) & .916 (.098) & .840 (.146)\\
\fastcgcstar{} & .953 (.058) & .828 (.163) & .917 (.097) & .803 (.178)\\
PCMCI+ & .943 (.063) & .879 (.123) & .891 (.116) & .857 (.138)\\
Bagged PCMCI+ & .964 (.044) & .864 (.140) & .933 (.076) & .842 (.155)\\
CEDAR & -- & .865 (.118) & -- & .816 (.141)\\
CUTS+ & .810 (.147) & .592 (.260) & -- & --\\
\midrule
Edge prevalence & .308 (.142) & -- & .148 (.054) & --\\
\bottomrule
\end{tabular}
\end{table*}

On sparse VAR systems, the fast algorithms closely matched their c-GC
counterparts in aggregate AUPRC but had lower F1 at the selected thresholds
(Table~\ref{tab:sparse-var-recovery}). \fastcgc{} achieved mean AUPRC 0.953
and F1 0.842, compared with 0.952 and 0.873 for \origcgc{}. The corresponding
\fastcgcstar{} values were 0.953 and 0.828, compared with 0.952 and 0.863 for
\origcgcstar{}. Bagged PCMCI+ had the highest aggregate and exact-lag AUPRC.
The paired aggregate-AUPRC differences from the c-GC baselines were small
(0.00081 and 0.00047), despite remaining significant after Holm correction.
For \fastcgc{}, performance improved with sample size and declined with
variable count and fitted order. Appendix~\ref{app:comparisons} gives the
contrasts and sensitivity summaries.

\subsubsection{Autocorrelated additive structural process}
\label{sec:scenario-autocorrelated-scm}

The second controlled process follows the autocorrelated additive structural
causal model used by \citet{debeire2024},
\begin{equation}
 X_t^j=a_jX_{t-1}^j
 +\sum_i c_i f_i\!\left(X_{t-\tau_i}^i\right)+\eta_t^j.
\label{eq:comparison-autocorrelated-scm}
\end{equation}
Cross-link lags are sampled from 1 through 5, coefficients from
\(\{\pm0.1,\pm0.2,\ldots,\pm0.5\}\), and innovation scales from \([0.5,2]\).
The reference cell has 10 variables, 500 observations, maximum autoregressive
coefficient 0.6, \(\lfloor1.5n\rfloor\) cross-links for \(n\) variables,
identity link functions, and Gaussian innovations.

Nine additional cells vary the variable count (5 or 20), sample count (250 or
1000), maximum autoregressive coefficient (0.2 or 0.95), link function, or
innovation family. The nonlinear option is
\(f(x)=x+5x^2\exp(-x^2/20)\); the non-Gaussian option uses standardized
Weibull innovations. Each cell has 50 requested seeds, and a cross-link-free
cell measures false positives. Contemporaneous links are disabled because the
eight methods do not return a common contemporaneous estimand.

\begin{table*}[t]
\centering
\caption{Autocorrelated-SCM recovery. AUPRC uses the 450 non-null graphs; F1
also includes the 50 cross-link-free graphs. Binary metrics for the four RCCP
algorithms use \((\alpha,\beta)=(0.01,0.005)\). Dashes mark an unavailable
metric or estimand.}
\label{tab:autocorrelated-recovery}
\footnotesize
\begin{tabular}{@{}lcccc@{}}
\toprule
Method & Aggregate AUPRC & Aggregate F1 & Exact-lag AUPRC & Exact-lag F1\\
\midrule
\origcgc{} & .844 (.113) & .717 (.265) & .544 (.176) & .570 (.239)\\
\fastcgc{} & .851 (.107) & .693 (.266) & .659 (.197) & .569 (.249)\\
\origcgcstar{} & .826 (.129) & .682 (.274) & .754 (.172) & .665 (.273)\\
\fastcgcstar{} & .836 (.116) & .658 (.282) & .771 (.153) & .650 (.282)\\
PCMCI+ & .844 (.110) & .739 (.265) & .777 (.148) & .726 (.265)\\
Bagged PCMCI+ & .888 (.090) & .735 (.266) & .839 (.121) & .729 (.266)\\
CEDAR & -- & .737 (.263) & -- & .720 (.262)\\
CUTS+ & .474 (.127) & .347 (.177) & -- & --\\
\midrule
Edge prevalence & .160 (.083) & -- & .032 (.017) & --\\
\bottomrule
\end{tabular}
\end{table*}

Bagged PCMCI+ had the highest aggregate and exact-lag AUPRC
(Table~\ref{tab:autocorrelated-recovery}). Longer series improved
\fastcgc{} AUPRC, while strong autocorrelation and nonlinear links reduced it.
Exact-lag recovery separated the RCCP conditioning sets more clearly than
aggregate recovery: mean exact-lag AUPRC was 0.659 for \fastcgc{} and 0.771
for \fastcgcstar{}. Figure~\ref{fig:exact-lag-auprc} summarizes both controlled
suites.

\begin{figure*}[t]
\centering
\includegraphics[width=0.88\textwidth]{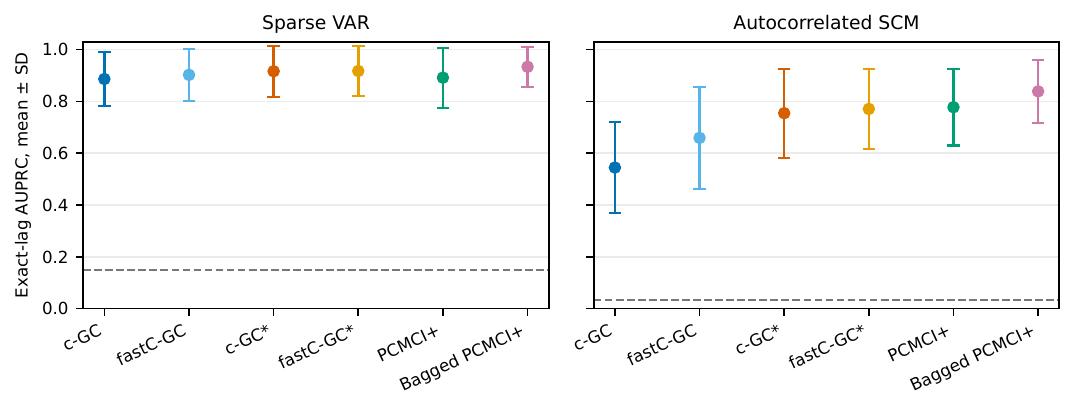}
\caption{Exact-lag AUPRC over directed cross-variable candidates at lags
1 through the fitted maximum; self-edges are excluded. Points show means,
and bars show sample standard deviations across graphs, not confidence
intervals for the means. Dashed lines show mean edge prevalence, which varies
across graphs. CEDAR is omitted because it supplies no continuous final-edge
scores; CUTS+ is omitted because it returns only an aggregate graph.}
\label{fig:exact-lag-auprc}
\end{figure*}

The fast methods' aggregate-AUPRC differences from their c-GC baselines were
0.00682 and 0.00947; only the latter remained significant after Holm correction.
Aggregate F1 was lower for both fast methods, whereas exact-lag F1 for
\fastcgc{} and \origcgc{} was nearly equal (0.569 and 0.570).
Appendix~\ref{app:comparisons} gives the full contrasts.

In the cross-link-free cell, the mean false positives per target for
\origcgc{}, \fastcgc{}, \origcgcstar{}, and \fastcgcstar{} were 0.140, 0.078,
0.070, and 0.054, respectively. PCMCI+, Bagged PCMCI+, CEDAR, and CUTS+
obtained 0.172, 0.072, 0.164, and 0.022.

\subsubsection{Lorenz--96}
\label{sec:scenario-lorenz96}

The third controlled process uses the Lorenz--96 system supplied with CUTS+
\citep{cheng2024},
\begin{equation}
 \frac{dX_i}{dt}=(X_{i+1}-X_{i-2})X_{i-1}-X_i+F,
\label{eq:lorenz96}
\end{equation}
with cyclic indices. The generator sets \(F=10\), uses \(\Delta t=0.1\),
discards 1000 samples, and adds independent Gaussian measurement noise with
standard deviation 0.1. The grid contains 10 and 20 variables, 250, 500, and
1000 retained observations, and 50 seeds per cell beginning at seed 20,000.
All methods use maximum lag 5. The target is the aggregate structural graph in
which each state depends on itself and its three Lorenz--96 neighbors; a unique
lag target is not available.

\begin{table}[t]
\centering
\caption{Aggregate Lorenz--96 recovery, mean (sample SD) across 300 systems.
Binary metrics for the four RCCP algorithms use
\((\alpha,\beta)=(0.01,0.005)\). A dash marks an unavailable ranking metric.}
\label{tab:lorenz-recovery}
\footnotesize
\begin{tabular}{@{}lcc@{}}
\toprule
Method & AUPRC & F1\\
\midrule
\origcgc{} & .692 (.091) & .596 (.130)\\
\fastcgc{} & .683 (.097) & .519 (.167)\\
\origcgcstar{} & .648 (.078) & .481 (.117)\\
\fastcgcstar{} & .667 (.081) & .404 (.156)\\
PCMCI+ & .596 (.104) & .480 (.131)\\
Bagged PCMCI+ & .781 (.079) & .430 (.151)\\
CEDAR & -- & .352 (.079)\\
CUTS+ & .948 (.090) & .857 (.113)\\
\midrule
Edge prevalence & .246 (.088) & --\\
\bottomrule
\end{tabular}
\end{table}

CUTS+ had the highest mean AUPRC (0.948) and F1 (0.857)
(Table~\ref{tab:lorenz-recovery}), reversing its lower ranking on the two
stochastic processes. Relative to their c-GC baselines, aggregate AUPRC was
0.00923 lower for \fastcgc{} and 0.01854 higher for \fastcgcstar{}; both
differences remained significant after Holm correction. F1 fell from 0.596 to
0.519 for \fastcgc{} and from 0.481 to 0.404 for \fastcgcstar{}. Thus, the
runtime gains again accompanied lower thresholded recovery performance.
The reference is the continuous-time structural graph; sampled predictive
relations need not coincide with this target.

\subsubsection{DREAM3 in-silico gene networks}
\label{sec:scenario-dream3}

DREAM3 Size 100 adds five in-silico gene-regulatory systems
\citep{prill2010}: two \emph{E. coli} networks and three yeast networks, each
with 100 genes. The challenge networks were derived from regulatory topologies
and simulated with continuous differential equations and Gaussian measurement
noise. We converted the licensed data to a common variable-by-time orientation.
Because the reference networks have no unique lag assignment, recovery is
scored only on the aggregate directed graph.

\begin{table*}[t]
\centering
\caption{Aggregate AUPRC on the five DREAM3 Size 100 networks. The networks are
distinct systems, not repeated seeds. Dashes mark unavailable metrics.}
\label{tab:dream3-networks}
\footnotesize
\begin{tabular}{@{}lccccc@{}}
\toprule
Method & Network 0 & Network 1 & Network 2 & Network 3 & Network 4\\
\midrule
\origcgc{} & .0463 & .0320 & .0550 & .0587 & .0666\\
\fastcgc{} & .0584 & .0355 & .0584 & .0599 & .0670\\
\origcgcstar{} & .0565 & .0402 & .0556 & .0722 & .0662\\
\fastcgcstar{} & .0691 & .0412 & .0588 & .0706 & .0663\\
PCMCI+ & .1242 & .1062 & .0937 & .0880 & .0699\\
Bagged PCMCI+ & .1410 & .1189 & .1225 & .0969 & .0836\\
CEDAR & -- & -- & -- & -- & --\\
CUTS+ & .0387 & .0215 & .0572 & .0485 & .0678\\
\bottomrule
\end{tabular}
\end{table*}

Bagged PCMCI+ had the highest mean AUPRC (0.1126), followed by PCMCI+
(0.0964); the fast methods achieved 0.0558 and 0.0611
(Table~\ref{tab:dream3-networks}). F1 was low for all four RCCP algorithms
(0.0616--0.0800), and both fast methods were below their c-GC counterparts.
With five heterogeneous networks, none of the eight paired AUPRC contrasts
remained significant after Holm correction. These results support a descriptive
comparison, with limited evidence for generalization across gene networks.

\subsubsection{Causal Chambers wind-tunnel measurements}
\label{sec:scenario-causal-chambers}

The first measured-data case study uses the Causal Chambers
\nolinkurl{wt_walks_v1} wind-tunnel recordings~\citep{gamella2025}. The
analysis uses the \nolinkurl{actuators_random_walk_1} experiment, 16 variables
from the authors' time-series case study, per-variable min--max normalization,
and maximum fitted lag 10. The variables cover hatch position, actuator loads
and potentiometers, input and output currents, upstream and downstream
pressures, fan speeds, a microphone, intake and ambient pressures, and two
amplifier signals. The distributed physical network is displayed as context,
not treated as exhaustive lag-resolved truth.
Assumptions~\ref{ass:causal-markov}--\ref{ass:present-sink} are not established
for these measurements. The
inferred networks are therefore interpreted descriptively rather than as
validated causal graphs.

Bagged PCMCI+, CEDAR, \origcgc{}, and \origcgcstar{} returned 4, 11, 8, and 4
cross-variable edges. CUTS+, \fastcgc{}, \fastcgcstar{}, and PCMCI+ returned
83, 5, 4, and 7 edges, respectively (Figure~\ref{fig:causal-chamber-graphs}).
The \fastcgc{} and \origcgc{} edge sets had
Jaccard similarity 0.625; \fastcgcstar{} and \origcgcstar{} returned the same four
aggregate edges. The five-edge \fastcgc{} graph and eight-edge \origcgc{} graph
confirm that the fast method must be evaluated as a separate algorithm. CUTS+
produced a much denser graph than the other methods, but the contextual network
does not support an efficacy ranking. For a concrete comparison, both
\fastcgc{} and \origcgc{} selected output fan speed \(\to\) upwind pressure
(nodes 11 and 9), linking a fan-speed measurement with pressure upstream in the
tunnel. This pair is absent as a direct edge from the contextual network.
Only \origcgc{} selected output current \(\to\) upwind pressure
(nodes 7 and 9). These are selected lagged associations; neither agreement
between methods nor physical proximity establishes a causal edge.
Figure~\ref{fig:causal-chamber-graphs} compares all eight methods.

\begin{figure*}[!t]
\centering
\begin{minipage}[c]{0.8\textwidth}
\centering
\graphpanel{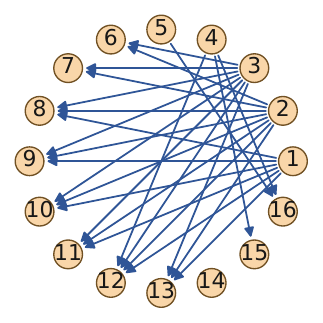}{a}{Physical network (context)}
\hfill
\graphpanel{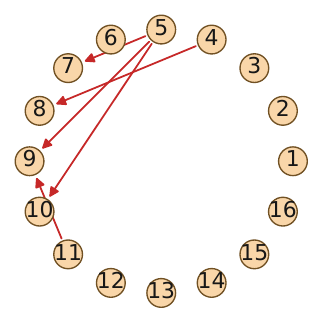}{b}{\fastcgc{}}
\hfill
\graphpanel{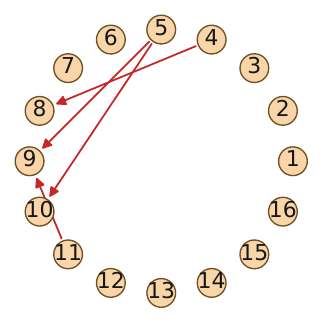}{c}{\fastcgcstar{}}\\[1.2ex]
\graphpanel{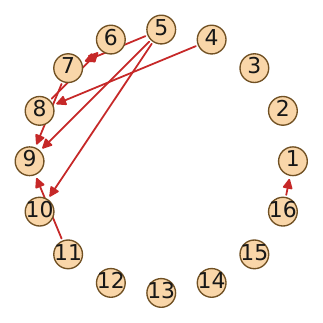}{d}{\origcgc{}}
\hfill
\graphpanel{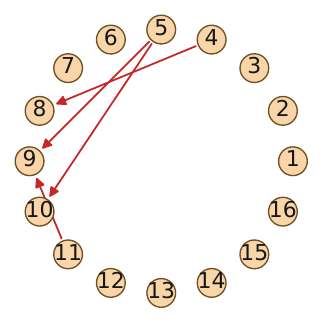}{e}{\origcgcstar{}}
\hfill
\graphpanel{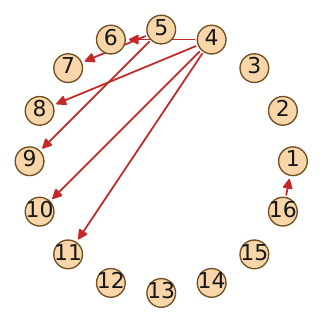}{f}{PCMCI+}\\[1.2ex]
\graphpanel{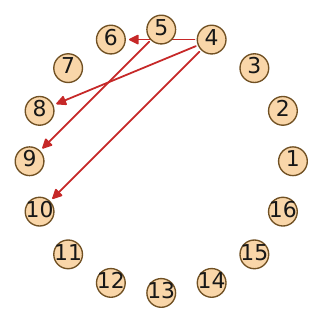}{g}{Bagged PCMCI+}
\hfill
\graphpanel{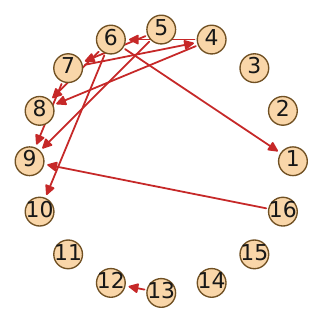}{h}{CEDAR}
\hfill
\graphpanel{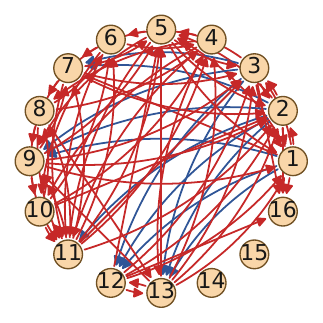}{i}{CUTS+}
\end{minipage}
\hfill
\chamberkey
\caption{Causal Chambers wind-tunnel graphs with a shared circular node layout.
Panel (a) shows the physical network for context, not exhaustive lag-resolved
truth. In (b)--(i), an arrow denotes a selected cross-variable link at any
positive lag from 1 through 10; self-edges are excluded.
\textcolor[HTML]{2F5597}{Blue} inferred arrows are present as directed edges
in (a); \textcolor[HTML]{C62828}{red} arrows are absent. Absence from this contextual network does not establish
a false positive. Width is uniform and does not encode a score; curvature
separates reciprocal links. The variable key applies to every panel.}
\label{fig:causal-chamber-graphs}
\end{figure*}

\subsubsection{CausalRivers discharge measurements}
\label{sec:scenario-causal-rivers}

The second measured-data case study uses CausalRivers discharge measurements
from eastern Germany and Bavaria~\citep{stein2025}. Following the authors'
connected-subgraph procedure, we select five five-station subgraphs per region
at evenly spaced positions in the sorted candidate list.
Raw 15-minute observations are rounded and averaged into 6-hour bins. Each
station is min--max normalized, internal missing values are linearly
interpolated, and incomplete rows at the beginning and end are removed. The
methods use maximum lag 5. River connectivity and station coordinates are
shown as geographic context rather than exhaustive causal truth.
Assumptions~\ref{ass:causal-markov}--\ref{ass:present-sink} are not established
for these measurements. The
inferred networks are therefore interpreted descriptively rather than as
validated causal graphs.

Across the ten subgraphs, the mean (SD) edge counts for \origcgc{}, \fastcgc{},
\origcgcstar{}, and \fastcgcstar{} were 10.9 (4.6), 9.4 (5.4), 12.2 (5.0), and
10.3 (4.8), respectively. PCMCI+, Bagged PCMCI+, CEDAR, and CUTS+ obtained
9.1 (2.4), 6.5 (2.1), 6.5 (2.7), and 7.3 (2.3). CEDAR's individual counts
ranged from 1 to 9. Mean Jaccard similarity was 0.795 (SD 0.149) for
\fastcgc{} versus \origcgc{} and 0.841 (SD 0.123) for \fastcgcstar{} versus \origcgcstar{}.
Figure~\ref{fig:causal-rivers-graphs} shows an illustrative Bavaria subgraph:
the first of the selected subgraphs in sorted candidate order, containing
station IDs 0, 94, 503, 504, and 505. Both \fastcgc{} and \origcgc{} selected
0 \(\to\) 503 (A \(\to\) C), which also appears in the supplied river-connectivity
network. Only \origcgc{} selected 0 \(\to\) 505 and 94 \(\to\) 504
(A \(\to\) E and B \(\to\) D). The methods selected different edge sets;
the available river network does not establish which inferred edges are causal.

\begin{figure*}[!t]
\centering
\begin{minipage}[c]{0.8\textwidth}
\centering
\graphpanel{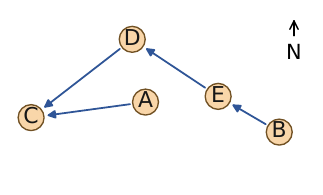}{a}{River connectivity (context)}
\hfill
\graphpanel{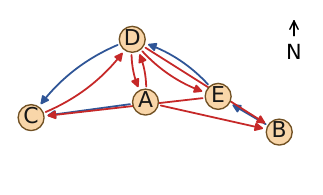}{b}{\fastcgc{}}
\hfill
\graphpanel{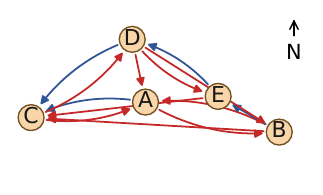}{c}{\fastcgcstar{}}\\[1.2ex]
\graphpanel{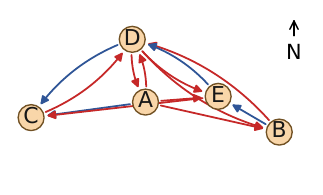}{d}{\origcgc{}}
\hfill
\graphpanel{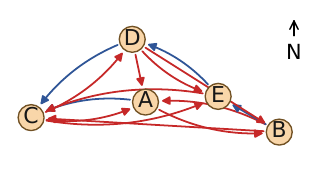}{e}{\origcgcstar{}}
\hfill
\graphpanel{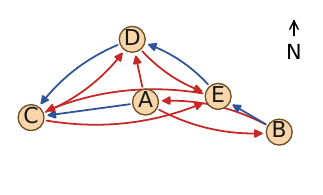}{f}{PCMCI+}\\[1.2ex]
\graphpanel{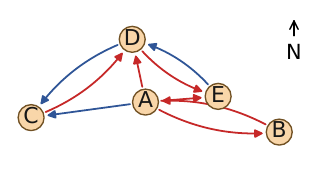}{g}{Bagged PCMCI+}
\hfill
\graphpanel{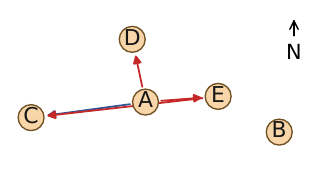}{h}{CEDAR}
\hfill
\graphpanel{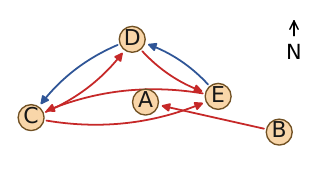}{i}{CUTS+}
\end{minipage}
\hfill
\riverskey
\caption{An illustrative CausalRivers Bavaria subgraph, with longitude
horizontal and latitude vertical in all panels; N marks north. In (a), arrows
follow the supplied downstream
river-connectivity network, which is contextual rather than exhaustive causal
truth. In (b)--(i), an arrow denotes a selected cross-station link at any
positive lag from 1 through 5; self-edges are excluded.
\textcolor[HTML]{2F5597}{Blue} inferred arrows are present as directed edges
in (a); \textcolor[HTML]{C62828}{red} arrows are absent. This distinction does
not establish causal correctness. Width does not encode a score, and curved
arrows separate reciprocal links. The station-ID key applies to every panel.}
\label{fig:causal-rivers-graphs}
\end{figure*}

\subsection{Limitations}
\label{sec:experimental-limitations}

The \(\beta=0.005\) choice is informed by empirical null rejection rates
in the examined settings. It does not establish uniform edgewise calibration,
conditional-test validity, or graph-wide error control. In particular, the null
study does not separately establish calibration for common-cause configurations
with marginal dependence and conditional independence. The drifting stress
test increased null rejection, and the probability of at least one false edge
reached 0.30. Independent null simulations after threshold selection are not reported;
they would provide a separate assessment of the selected threshold. The original \(\beta=0.001\) rule
produces empty fast-method graphs when \(L<1000\).

The ridge sensitivity grid spans four orders of magnitude, but it does not
establish invariance to regularization outside the tested processes. Likewise,
the maximum-lag experiment documents under- and over-specification around one
second-order sparse VAR; it does not supply a data-driven lag-selection rule.

The Bagged PCMCI+ score correction assumes that the positive-lag graph has only
two marks: a directed edge and no edge. Additional mark types would require
edge frequencies computed from the individual bootstrap graphs.

CEDAR can be compared only through its final binary graphs; converting its
intermediate p-values or adjacency into a ranking would not measure the final
selection step. Its first-lag policy may also miss additional true lags for the
same ordered variable pair.

The controlled study is limited to observed variables, stationary processes,
the declared innovation and link-function families, and the specified lag orders.
Shared latent causes, contemporaneous effects, structural breaks, missing
observations, and data-driven lag selection require separate experiments.
DREAM3 contributes five heterogeneous networks, which limits inference from
its across-network mean.

The reported rankings depend on both the data-generating process and the
evaluation target. The Lorenz--96 and DREAM3 reference graphs specify
structural relations without a unique sampled lag, while the measured-data
networks provide context rather than exhaustive truth. Preprocessing,
including river-data interpolation and temporal aggregation, can also alter
lagged associations.

The runtime comparison measures end-to-end implementations on one host and
one finite grid. It does not isolate the contributions of FFT evaluation,
shared matrix calculations, and estimator changes. The reported experiments include no computational ablation
or peak-memory measurement, so the observed speedups cannot be
attributed to individual components or assumed to transfer to other hardware.

\FloatBarrier

\section{Conclusion}
\label{sec:conclusion}

\fastcgc{} and \fastcgcstar{} combine the RCCP two-test rule and established
conditioning sets with standardized correlation estimates, ridge residuals,
and complete FFT-evaluated circular shift orbits. Shared precision calculations
and Schur downdates yield substantial computational savings: on the tested
single-thread grid, the methods were respectively 68--277 and 192--886 times
faster than the compared c-GC implementations.

These gains accompany a change in estimator. Across four reference-graph
suites, aggregate mean AUPRC differed from the corresponding c-GC baseline by
at most 0.019, but F1 was lower for both fast methods in every suite at the
selected thresholds. The methods are therefore useful computational alternatives
when ranking performance and runtime are priorities; agreement of selected
graphs must be assessed separately. Bagged PCMCI+ led three suites in mean
AUPRC, while CUTS+ led on Lorenz--96, underscoring dependence on the dynamical
regime.

The algebraic results justify the calculations. The population adjacency
result for \fastcgcstar{} additionally requires the stated graphical
assumptions and does not establish consistency of the implemented fixed-ridge
estimator. Empirical null rejection rates informed \(\beta=0.005\), but do
not establish general test validity or graph-wide error control. The measured
data illustrate differences between inferred graphs without identifying the
causally correct output. Priorities for further work are conditional-null
validation, multiplicity control, and computational ablations with memory
measurements.

\section*{Code availability}
The source code and experiments are available at\par\noindent
\url{https://github.com/adedayoas91/fastC-GC}.

\appendix
\section{Supporting algebra and diagnostics}
\label{app:algebra}

The following proofs support the statements in Section~\ref{sec:method}.
The diagnostic transform below is recorded by the implementation but does not
enter the graph decision.

\subsection{FFT orbit identity}
\begin{proof}
Expanding the inverse transform gives
\[
 \frac{1}{L}\sum_{k=0}^{L-1}
 \left(\sum_{u=0}^{L-1}x_u e^{-2\pi iku/L}\right)
 \left(\sum_{v=0}^{L-1}y_v e^{2\pi ikv/L}\right)
 e^{2\pi iks/L}.
\]
Interchanging sums, the inner sum over \(k\) is \(L\) when
\(v-u+s\equiv0\pmod L\) and zero otherwise. What remains is
\(\sum_u x_u y_{(u-s)\bmod L}=c_s\). A real FFT stores only the
nonredundant half-spectrum but computes the same inverse transform.
\end{proof}

\subsection{Ridge existence and perturbation}
\begin{proof}
For nonzero \(v\),
\(v^\top(S+\lambda I)v=v^\top Sv+\lambda\|v\|_2^2>0\).
Diagonalizing \(S=U\operatorname{diag}(\mu_k)U^\top\), the eigenvalues of
\((S+\lambda I)^{-1}\) are \(1/(\mu_k+\lambda)\). If all \(\mu_k>0\), the
eigenvalues of the difference from \(S^{-1}\) have magnitude
\(\lambda/[\mu_k(\mu_k+\lambda)]\), maximized at \(\mu_{\min}\).
\end{proof}

\subsection{Affine scale invariance}
\begin{proof}
The standardized expanded matrix becomes
\(\widetilde Z'=D\widetilde Z\), where \(D\) is diagonal with entries
\(\operatorname{sign}(c_a)\) repeated across lag blocks. Since \(D^2=I\),
\[
\begin{aligned}
 S'&=DSD,\\
 S'+\lambda I&=D(S+\lambda I)D,\\
 P'&=DPD.
\end{aligned}
\]
Absolute marginal and precision correlations are therefore unchanged. For a
pair \(e\) and controls \(B\),
\[
\begin{aligned}
 R'_e
 &=D_e\widetilde Z_e-D_eS_{eB}D_B\\
 &\quad{}\times\bigl[D_B(S_{BB}+\lambda I)D_B\bigr]^{-1}
   D_B\widetilde Z_B\\
 &=D_eR_e.
\end{aligned}
\]
Sign changes leave the absolute residual correlations and their absolute
circular shift orbits unchanged.
\end{proof}

\subsection{Optional precision diagnostics}
\label{app:diagnostics}
Set \(q=q_{ij}^A\). This precision quantity is cheap to compute for many pairs. With
\(d=|A|-2\) controls, it maps to the conventional correlation statistic
\begin{equation}
\begin{aligned}
 \nu&=\max(L-d-2,1),\\
 t&=q\sqrt{\frac{\nu}{1-q^2}},\\
 p_{\mathrm{scr}}&=2\,\Pr(T_\nu\geq t).
\end{aligned}
\label{eq:diagnostic-p}
\end{equation}
This Student-\(t\) calculation assumes independent Gaussian observations and an
unregularized partial correlation. It is therefore not a calibrated \(p\)-value
for general autocorrelated ridge data. We report \(q_{ij}^A\) and
\(p_{\mathrm{scr}}\) as diagnostics but do not use them to select candidates.
Every allowed conditional pair receives ridge-residual recovery and complete
circular shift calibration.

\section{Detailed benchmark comparisons}
\label{app:comparisons}

These numerical comparisons supplement the principal results in
Section~\ref{sec:results}. Differences are descriptive unless an inferential
result is stated. Statistical significance does not imply estimator equivalence.

\subsection{Sparse VAR}
Bagged PCMCI+ had the highest mean aggregate and exact-lag AUPRC
(Table~\ref{tab:sparse-var-recovery}). CEDAR had aggregate F1 0.8653 and
exact-lag F1 0.8160; its aggregate precision was 0.9748 and recall was 0.7952.
The two fast algorithms and their c-GC baselines were close on the aggregate
graph. At the lag-resolved level,
\fastcgc{} was 0.0160 above \origcgc{} in mean AUPRC, whereas \fastcgcstar{}
and \origcgcstar{} differed by 0.0009. These exact-lag differences are descriptive; the
declared paired tests concern aggregate AUPRC.

The \fastcgc{} aggregate score improved from 0.9027 at 250 observations to
0.9882 at 1000 observations. It declined from 0.9763 with five variables to
0.9252 with 20, and from 0.9635 at order one to 0.9415 at order five. Gaussian
and Laplace innovations differed by 0.0028 in the pooled means. Per-graph AUROC
correlations were 0.9965 for \fastcgc{} versus \origcgc{} and 0.9970 for
\fastcgcstar{} versus \origcgcstar{}; the corresponding AUPRC correlations were 0.9845
and 0.9870.

In the paired aggregate analysis, \fastcgc{} minus \origcgc{} had mean AUPRC
difference 0.00081 (95\% bootstrap interval \([0.00043,0.00121]\), Holm
\(p=0.00039\)); \fastcgcstar{} minus \origcgcstar{} differed by 0.00047
(\([0.00013,0.00080]\), Holm \(p=0.0153\)). Both fast algorithms were about
0.010 above PCMCI+ and 0.143 above CUTS+, whereas Bagged PCMCI+ was about 0.011
higher than either fast algorithm (all four external contrasts Holm
\(p<2\times10^{-12}\)). The small baseline differences are evidence of similar
rankings, not numerical identity.
At the selected thresholds, aggregate F1 was 0.8421 for \fastcgc{} versus
0.8726 for \origcgc{}, and 0.8282 for \fastcgcstar{} versus 0.8625 for
\origcgcstar{}. The exact-lag differences had the same direction. Thus the
near-equality of the AUPRC rankings does not imply equality of the selected
edge sets.

\subsection{Autocorrelated SCM: recovery and sensitivity}
Bagged PCMCI+ again had the highest aggregate and exact-lag AUPRC
(Table~\ref{tab:autocorrelated-recovery}). CEDAR's aggregate F1 was 0.7368,
with precision 0.8410 and recall 0.6669; its exact-lag F1 was 0.7200. In the
reference cell, aggregate AUPRC was 0.8942 for \fastcgc{}, 0.8892 for
\fastcgcstar{}, 0.8912 for \origcgc{}, and 0.8853 for \origcgcstar{}. PCMCI+,
Bagged PCMCI+, and CUTS+ obtained 0.8751, 0.9091, and 0.4946, respectively.
Increasing the series length from 250 to 1000 raised \fastcgc{} AUPRC
from 0.8195 to 0.9389. Raising the maximum autoregressive coefficient from 0.20
to 0.95 reduced it from 0.8889 to 0.7750, and the mixed nonlinear-link cell
reached 0.6894.

The exact-lag results separate the two RCCP conditioning sets more clearly.
Mean AUPRC was 0.5445 for \origcgc{} and 0.6590 for \fastcgc{}, compared with 0.7543
for \origcgcstar{} and 0.7707 for \fastcgcstar{}. Bagged PCMCI+ reached 0.8385. Figure
\ref{fig:exact-lag-auprc} shows both controlled suites and includes the
lag-specific \origcgc{}/\origcgcstar{} results obtained from their recorded
\(p\)-values.

\subsection{Autocorrelated SCM: paired contrasts}
Per-graph aggregate AUROC correlations were 0.9955 for \fastcgc{} versus \origcgc{}
and 0.9956 for \fastcgcstar{} versus \origcgcstar{}; aggregate AUPRC correlations were
0.9677 and 0.9717. The paired aggregate differences were 0.00682 for
\fastcgc{} minus \origcgc{} (Holm \(p=0.197\)) and 0.00947 for \fastcgcstar{} minus
\origcgcstar{} (Holm \(p=0.00366\)). \fastcgc{} was 0.0074 above PCMCI+ (Holm
\(p=0.0174\)); the \fastcgcstar{}--PCMCI+ difference of \(-0.0079\) did not
remain significant after correction. Bagged PCMCI+ was 0.0367 and 0.0520 above
the two fast algorithms, while both fast algorithms were more than 0.36 above
CUTS+.
The reported aggregate F1 values were 0.6927 for \fastcgc{} and 0.6579 for
\fastcgcstar{}, compared with 0.7169 and 0.6820 for their respective c-GC
comparators. At exact lags, \fastcgc{} and \origcgc{} were nearly equal
(0.5685 and 0.5700), while \fastcgcstar{} was 0.0145 below
\origcgcstar{}.

\subsection{Lorenz--96}
Table~\ref{tab:lorenz-recovery} shows that CUTS+ had the highest mean AUPRC
(0.9482) and F1 (0.8575), followed by Bagged PCMCI+ in AUPRC (0.7809).
CEDAR returned mean F1 0.3525, with precision 0.5067 and recall 0.2834.
\fastcgc{} was 0.00923 below \origcgc{}, while
\fastcgcstar{} was 0.01854 above \origcgcstar{}; both aggregate differences remained
significant after Holm correction. The corresponding AUROC correlations were
0.9934 and 0.9973. PCMCI+ was 0.0870 and 0.0712 below the fast algorithms;
Bagged PCMCI+ was 0.0981 and 0.1140 above them; and CUTS+ was 0.2655 and 0.2813
above them (all Holm \(p<10^{-13}\)). The reversal between CUTS+ here and its
lower values on the two stochastic processes shows that the method ranking is
specific to the dynamical regime.
At the selected thresholds, \fastcgc{} and \fastcgcstar{} had mean F1 values
of 0.5187 and 0.4037, compared with 0.5957 and 0.4806 for \origcgc{} and
\origcgcstar{}.

\subsection{DREAM3}
Bagged PCMCI+ had the highest mean AUPRC, 0.1126, followed by PCMCI+ at 0.0964
(Table~\ref{tab:dream3-networks}).
The means were 0.0517 for \origcgc{}, 0.0558 for \fastcgc{}, 0.0581 for \origcgcstar{},
0.0611 for \fastcgcstar{}, and 0.0468 for CUTS+. The per-network AUROC
correlations were 1.0000 for \fastcgc{} versus \origcgc{} and 0.9998 for
\fastcgcstar{} versus \origcgcstar{}. With only five heterogeneous networks, none of
the eight paired AUPRC contrasts remained significant after Holm correction
(all Holm \(p=0.500\)). These data support a descriptive comparison rather than
a population-level ranking.
With the selected thresholds, mean F1 was 0.0800 for \origcgc{}, 0.0728 for
\fastcgc{}, 0.0658 for \origcgcstar{}, and 0.0616 for \fastcgcstar{}.
CEDAR's mean F1 was 0.0665 (SD 0.0511), with values from 0.0205 to 0.1294
across the five networks. Its mean precision was 0.1835 and mean recall was
0.0428.

\balance
\bibliographystyle{cas-model2-names}
\bibliography{references}

\end{document}